\documentclass[12pt]{amsart}

\usepackage[hmargin=0.8in,twosideshift=0in,height=8.6in]{geometry}
\usepackage{amssymb,amsthm}
\usepackage{delarray,verbatim}
\usepackage{natbib}
\renewcommand{\cite}{\citet}

\newtheorem{thm}{Theorem}[section]
\newtheorem{lemma}{Lemma}[section]
\newtheorem{remark}{Remark}[section]
\newtheorem{cor}{Corollary}[section]

\newtheorem{proposition}{Proposition}[section]

\numberwithin{equation}{section} 
\newcommand{\pv}{\pi^*}

\newcommand{\E}{\mathbb{E}}
\newcommand{\R}{\mathbb{R}}
\renewcommand{\P}{\mathbb{P}}

\newcommand{\la}{\lambda}
\newcommand{\noi}{\noindent}

\newcommand{\ML}{m}

\title[Optimal Investment Strategy to Minimize Occupation Time]{Optimal Investment Strategy to Minimize Occupation Time}
\author{Erhan Bayraktar }
\address[Erhan Bayraktar]{Department of
  Mathematics, University of Michigan, Ann Arbor, MI 48109}
\email{erhan@umich.edu}
\thanks{E. Bayraktar is supported in part by the National Science Foundation, under grant DMS-0604491. }

\author{Virginia R. Young}
\address[Virginia R. Young]{Department of
  Mathematics, University of Michigan, Ann Arbor, MI 48109}
\email{vryoung@umich.edu}


\begin{document}

\begin{abstract}
We find the optimal investment strategy to minimize the expected time that an individual's wealth stays below zero, the so-called {\it occupation time}.  The individual consumes at a constant rate and invests in a Black-Scholes financial market consisting of one riskless and one risky asset, with the risky asset's price process following a geometric Brownian motion.  We also consider an extension of this problem by penalizing the occupation time for the degree to which wealth is negative.

\end{abstract}
\maketitle

\section{Introduction}

We find the optimal investment strategy to minimize the expected time that an individual's wealth remains negative during that person's life.  The time the wealth process spends below a given level (0, in this paper) is known as the {\it occupation time}.  \cite{Akahori}, \cite{Takacs}, \cite{DoneyYor}, and \cite{Pechtl} find the distribution of the occupation time for Brownian motion with drift.  \cite{Linetsky} and \cite{Hugonnier} derive pricing formulas of various types of derivatives based on the occupation time.  To our knowledge, ours is the first study in controlling the occupation time.

We were motivated to study this problem in order to give individuals advice for investing after they become bankrupt.  One way to do this is to set an arbitrarily (large) negative ruin level and apply the investment strategy given by minimizing the probability of hitting that ruin level.  \cite{by3} show that the optimal investment strategy is independent of the ruin level.  Moreover, they show that if an individual were to minimize the expectation of a non-increasing function of minimum wealth, of which probability of ruin is an example, then the optimal investment strategy is independent of that function of wealth.

Another way to address the problem of how to invest when wealth is negative is to minimize the expected occupation time below zero during one's life.  When minimizing the probability of lifetime ruin, the ``game'' ends as soon as ruin occurs.  However, even after ruin occurs, the individual is still alive and requires an investment strategy.  In this paper, we consider the problem of minimizing the expected length of time that one's wealth is negative.

Our objective is to give investment advice to retirees.  This is driven by the fact that retired Americans have an ever increasing risk of having their living expenses exceed their financial resources.  This risk is due to the increasing longevity of individuals and inadequate retirement savings.  Moreover, individuals have to shoulder more financial risk now because of the continuing shift  in pension plans towards defined contribution.  To determine an investment policy, the individual needs to select an optimization criterion. The most common approach is to choose a utility function and  determine a policy that maximizes oneÕs expected discounted utility of consumption and bequest.  However, when maximizing lifetime utility of consumption and bequest under the commonly used utility functions, the optimal investment strategy does not depend on the hazard rate.  In giving investment advice to retirees, one expects that the advice should depend on the likelihood of dying (and she should be more conservative if this likelihood is higher), and our optimization criterion leads to such advice.  
Also, maximizing utility  is arguably subjective, while minimizing the probability of lifetime ruin or minimizing the expected occupation time might prove easier for retirees to understand.

The layout of the paper is as follows.  In Section \ref{sec:occ-time}, we present and solve the problem of minimizing expected occupation time.  In Section \ref{sec:fin-model}, we present the financial market in which the individual can invest, and we formulate the problem of minimizing expected occupation time.  In Section \ref{sec:verf}, we provide a verification theorem for the minimum expected occupation time, which we use to obtain the value function in Section \ref{sec:FBP} via a related free-boundary problem.  In Section \ref{sec:prop}, we present properties of the minimum occupation time and the corresponding optimal investment strategy.  In Section \ref{sec:exten}, we discuss an extension of minimizing expected occupation time.

\section{Minimizing Expected Occupation Time}\label{sec:occ-time}

In Section \ref{sec:fin-model}, we describe the financial market in which the individual can invest her wealth, and we formulate the problem of minimizing expected occupation time.  We also provide a verification theorem for the expected occupation time in Section \ref{sec:verf}; then, in  Section \ref{sec:FBP}, we construct the minimum expected occupation time as the convex dual of a solution of a related free-boundary problem.

\subsection{Financial model}\label{sec:fin-model}

We consider an individual aged with future lifetime described by the random variable $\tau_d$.  Suppose $\tau_d$ is an exponential random variable with parameter $\la$, also referred to as the hazard rate; in particular, $\E(\tau_d) = 1/\la$.

We assume that the individual consumes wealth at a constant {\it net} rate of $c$; this rate might be given in real or nominal units.  We say that the rate $c$ is a net rate because it is the rate of consumption offset by any income. One can interpret $c$ as the minimum net consumption level below which the individual cannot (or will not) reduce her consumption further; therefore, the minimum expected occupation time that we compute gives a lower bound for the expected occupation time under any consumption function bounded below by $c$.

The individual can invest in a riskless asset, which earns interest at the rate $r \ge 0$.  Also, she can invest in a risky asset whose price satisfies
\begin{equation}
dS_t = \mu S_t dt + \sigma S_t dB_t, \quad S_0 = S > 0,
\end{equation}
in which $\mu > r$, $\sigma > 0$, and $B$ is a standard Brownian motion with respect to a filtration $\mathbb{F} = \{ \mathcal{F}_t \}_{t \geq 0}$ of a probability space $(\Omega, \mathcal{F}, \P)$.  We assume that $B$ is independent of $\tau_d$, the random time of death of the individual.  If $c$ is given as a real rate of consumption (that is, after inflation), then we also express $r$ and $\mu$ as real rates.

Let $\pi_t$ denote the amount invested in the risky asset at time $t$, and let $\pi$ denote the investment strategy  $\{\pi_t\}_{t \geq 0}$.  We say that a strategy $\pi$ is {\it admissible} if the process $\pi$ is adapted to the filtration $\mathbb{F}$ and if $\pi$ satisfies the condition $\int_0^t \pi_s^2 \, ds < \infty$, almost surely, for all $t \ge 0$.  The wealth dynamics of the individual for a given admissible strategy $\pi$ are given by
\begin{equation}\label{eq:wealth}
dW_t = (r W_t + (\mu - r) \pi_t - c) dt + \sigma \pi_t dB_t, \quad
W_0 = w \ge 0.
\end{equation}

By {\it lifetime occupation time}, we mean the length of time that the individual's wealth is negative, before she dies.  One could also consider the time that wealth is below some arbitrary level, not necessarily $0$, but for ease of presentation, we choose the level to be $0$.  Define the process $A = \{A_t\}_{t \ge 0}$ by
\begin{equation}\label{eq:A}
A_t = A_0 + \int_0^t {\bf 1}_{\{W_s < 0\}} \, ds, \quad A_0 = a \ge 0.
\end{equation}
Thus, we wish to minimize $\E^{w,a} \left(A_{\tau_d} \right)$, in which $\E^{w, a}$ denotes expectation conditional on $W_0 = w$ and $A_0 = a$.  By allowing $A_0$ to be non-zero, we allow the individual to have been in bankruptcy for a period of time prior to the beginning of the game.

There is a problem with the goal of minimizing $\E^{w,a} \left(A_{\tau_d} \right)$.  Indeed, $0 \le \E^{w,a} \left(A_{\tau_d} \right) \le a + \E(\tau_d) = a + 1/\la$, and we expect the minimum $\E^{w, a} \left(A_{\tau_d} \right)$ to be a convex, non-increasing function of initial wealth $w$.  However, there is no bounded, convex, non-increasing function defined on the reals, other than a constant function.  Therefore, we modify the problem as follows:  Define the minimum wealth process $Z = \{Z_t\}_{t \ge 0}$ by
\begin{equation}
Z_t = \min_{0 \le s \le t} W_s.
\end{equation}
For $L > 0$, define the hitting time $\tau_L = \inf\{ t \ge 0: W_t \le -L \}$, and define the value function $M_L$ by
\begin{equation}\label{eq:M-L}
M_L(w, a) = \inf_\pi \E^{w, a} \left[ A_{\tau_d} \, {\bf 1}_{\{Z_{\tau_d} > -L \}} + \left( A_{\tau_L} + \frac{1}{\la} \right)  {\bf 1}_{\{Z_{\tau_d} \le -L \}} \right].
\end{equation}
Here, we take the infimum over admissible investment strategies.

For large values of $L > 0$, the control problem associated with $M_L$ approximates the problem of minimizing lifetime occupation time.  Indeed, if lifetime wealth stays above $-L$, which is quite likely for $L$ large, then the payoff is the occupation time.  If wealth falls below $-L$, then we suppose that the individual will have negative wealth for the rest of her life, and we end the game, with payoff equal to the occupation time up to the hitting time $\tau_L$ plus $\E(\tau_d) = 1/\la$.

Note that for $w \le -L$, $M_L(w, a) = a + 1/\la$, and for $w \ge c/r$, $M_L(w, a) = a$.  The latter holds because if $w \ge c/r$, then the individual can place all her wealth in the riskless asset and wealth will never go below $w \ge c/r$, much less reach 0.  Therefore, it remains for us to determine $M_L$ on $(-L, c/r)$.  The verification theorem in the next section will help us with that goal.

\subsection{Verification theorem for $M_L$}\label{sec:verf}

\begin{thm}\label{thm:verf}
Suppose $m: \R \to \R$ is a bounded, continuous function that satisfies the following conditions:
\item{$(i)$} $m$ is convex on $(-L, c/r)$ and lies in $C^2((-L,c/r)),$ except at $0;$
\item{$(ii)$} $m \in C^1(\R);$
\item{$(iii)$} $m(w) = 1/\la$ for all $w \le -L;$
\item{$(iv)$} $m(w) = 0$ for $w \ge c/r;$
\item{$(v)$} $m$ solves the following Hamilton-Jacobi-Bellman $($HJB$)$ equation on $(-L, c/r):$
\begin{equation}\label{eq:HJB}
\la m(w) = {\bf 1}_{\{w < 0\}} + (rw - c) m'(w) + \min_{\pi} \left[(\mu-r)\pi m'(w) + \frac{1}{2}\sigma^2 \pi^{2} m''(w) \right].
\end{equation}
Then, the value function $M_L$ on $\R$ defined by $(\ref{eq:M-L})$ is given by
\begin{equation}\label{eq:verf}
M_L(w, a) = m(w) + a,
\end{equation}
and the optimal investment strategy $\pi^*$ on $(-L, c/r)$ is given in feedback form by 
\begin{equation}\label{eq:pi}
\pi^*_t = - \frac{\mu - r}{\sigma} \, \frac{m'(W^*_t)}{m''(W^*_t)},
\end{equation}
in which $W^*$ is the optimally controlled wealth.
\end{thm}

\begin{proof}
Assume that $m$ satisfies the conditions specified in the statement of this theorem. Let $N$ denote a Poisson process with rate $\lambda$ that is independent of the standard Brownian motion $B$ driving the wealth process. The occurrence of a jump in the Poisson process represents the death of the individual.

Let $\pi: \R \to \R$ be a function, and let $W^\pi$, $A^\pi$, and $Z^\pi$ denote the wealth, occupation time, and minimum wealth, respectively, when the individual follows the investment policy $\pi_t = \pi(W_t)$.  Assume that this investment policy is admissible.

Define $\bar \R = \R \cup \{\infty\}$ to be the one-point compactification of $\R$.  The point $\infty$ is  the ``coffin state."  The wealth process is killed (and sent to the coffin) as soon as the Poisson process jumps (that is, when the individual dies), and we assign $W^\pi_{\tau_d} = \infty$; recall that $\tau_d$ is the random time of death.  Given a function $g$ on $\R$, we extend it to $\bar \R$ by $g(\infty) = 0$.

Define the stopping time $\tau = \tau_d \wedge \tau_{c/r} \wedge \tau_L$, in which $\tau_{c/r} = \inf\{t \ge 0: W_t \ge c/r \}$; recall that $\tau_L = \inf \{t \ge 0: W_t \le -L\}$.  Also, define $\tau_n = \inf \{t \ge 0:  \int_0^t \pi^2_s \, ds \ge n \}$.

If $W_0 = w \ge c/r$, then the individual can invest her wealth in the riskless asset and finance the cost of her consumption, which is no greater than $\int_{0}^{\infty} c \, e^{-rt} dt = c/r$. Therefore, with this strategy, her wealth will be at least $c/r$ at the time of her death.  In other words, if $W_0 = w \ge c/r$, then  $Z_{\tau_d} \ge c/r$, which implies that $M_L(w, a) = a$ for $w \ge c/r$.  Similarly, if $W_0 = w \le - L$, then $Z_{\tau_d} \le -L$, which implies that $M_L(w, a) = a + 1/\la$ for $w \le -L$.  From these observations it follows that
\begin{equation}\label{eq:newM-L}
M_L(w, a) = \inf_\pi \E^{w, a} \left[ A_{\tau} \, {\bf 1}_{\{Z_{\tau} > -L \}} + \left( A_{\tau} + \frac{1}{\la} \right) {\bf 1}_{\{Z_{\tau} \le -L \}} \right] = \inf_\pi \E^{w, a} \left[ A_{\tau} + \frac{1}{\la} {\bf 1}_{\{Z_{\tau} \le -L \}} \right].
\end{equation}

By applying a generalized It\^o's formula for convex functions (see Theorem 6.22, page 214 of \cite{kn:karat}) to $m(w) + a$ for $w \in (-L, c/r)$ and $a \ge 0$, we have
\begin{equation}\label{eq:Ito}
\begin{split}
&m(W^\pi_{t \wedge \tau \wedge \tau_n}) + A^\pi_{t \wedge \tau \wedge \tau_n} 
\\&= m(w) + a +  \int_0^{t \wedge \tau \wedge \tau_n}  \, dA_s 
 \quad +\int_0^{t \wedge \tau \wedge \tau_n}  \left( (r W^\pi_s + (\mu-r) \pi_s - c) \, m'(W^\pi_s) + \frac{1}{2} \,\sigma^2 \, \pi_s^2 \, m''(W^\pi_s) \right)ds \\
& \quad - \la \int_0^{t \wedge \tau \wedge \tau_n}  m(W^\pi_s) ds + \int_0^{t \wedge \tau \wedge \tau_n} m'(W^\pi_s) \, \sigma \, \pi_s \, dB_s - \int_0^{t \wedge \tau \wedge \tau_n} m(W^\pi_s) d(N_s - \la s). 
\end{split}
\end{equation}
Here, we used the fact if the Poisson processes jumps at time $s$ the change in the value of the process $m(W_t^{\pi})+A^{\pi}_t$ at time $s$ is $m(W_s^{\pi})-m(W_{s-}^{\pi})=-m(W_{s-}^{\pi})$, in which the last equality follows since $m(\infty)=0$. 
In the derivation of \eqref{eq:Ito} we also used the fact that $\int_0^{t}m(W_{s-}^{\pi}) ds=\int_0^{t}m(W_{s}^{\pi}) ds$ which follows thanks to the fact that the Lebesgue measure is diffuse.

It follows from the definition of $A$ in (\ref{eq:A}) that $dA_s = {\bf 1}_{\{ w < 0\}} \, ds$.  The definition of $\tau_n$ implies that the expectation of the fourth integral is 0.  Moreover, the expectation of the fifth integral is zero since $m$ is bounded; see, for example, \cite{MR82m:60058}.

Thus, we have
\begin{align}\label{eq:3.9}
& \E^{w, a} \left[ m \left(W^\pi_{t \wedge \tau \wedge \tau_n} \right) + A^\pi_{t \wedge \tau \wedge \tau_n} \right] = m(w) + a - \E^{w, A} \left[ \int_0^{t \wedge \tau \wedge \tau_n} \la \, m(W^\pi_s) \, ds \right]  \nonumber \\
& + \E^{w, A} \left[ \int_0^{t \wedge \tau \wedge \tau_n} \left( {\bf 1}_{\{W^\pi_s < 0\}} + (rW^\pi_s + (\mu-r) \pi_s - c) \, m'(W^\pi_s) + \frac{1}{2} \,\sigma^2 \, \pi_s^2 \, m''(W^\pi_s) \right) ds \right] \nonumber \\
& \geq m(w) + a, 
\end{align}
in which the inequality follows from assumption (v) of the theorem.  Because $m$ is bounded and because $A_t \le a + \tau_d$ almost surely, it follows from the dominated convergence theorem that
\begin{equation}\label{eq:sub-mart}
\E^{w, a} \left[ m(W^\pi_{t \wedge \tau}) + A^\pi_{t \wedge \tau} \right] \geq m(w) +  a.
\end{equation}
Inequality (\ref{eq:sub-mart}) shows that $\left\{ m(W^\pi_{t \wedge \tau}) + A^\pi_{t \wedge \tau} \right\}_{t \ge 0}$ is a sub-martingale for any admissible strategy $\pi$.

>From $m(c/r) = 0 = m(W^\pi_{\tau_d})$ and $m(-L) = 1/\la$, it follows that
\begin{align}\label{eq:m}
m(W^\pi_\tau) +  A^\pi_\tau &= \left( m \left(W^\pi_{\tau_d} \right) + A^\pi_{\tau_d} \right) {\bf 1}_{\left\{\tau_d < \left( \tau_{c/r} \wedge \tau_L \right) \right\}}
+  \left( m(c/r) + A^\pi_{\tau_{c/r}} \right) {\bf 1}_{ \left\{\tau_{c/r} < \left( \tau_d \wedge \tau_L \right) \right\}} \nonumber \\
& \quad + \left( m (-L) + A^\pi_{\tau_L} \right) {\bf 1}_{ \left\{\tau_L < \left( \tau_d \wedge \tau_{c/r} \right) \right\}}  \nonumber \\
& = A^\pi_\tau + \frac{1}{\la} \, {\bf 1}_{ \left\{ Z^\pi_\tau \le -L \right\}}.
\end{align}
By taking the expectation of both sides of (\ref{eq:m}), we obtain
\begin{equation}\label{eq:3.12}
\E^{w, a} \left[A^\pi_\tau + \frac{1}{\la} \, {\bf 1}_{ \left\{ Z^\pi_\tau \le -L \right\}} \right] = \E^{w, a} \left[ m(W^\pi_{\tau}) + A^\pi_{\tau} \right] \ge m(w) + a.
\end{equation}
The inequality in (\ref{eq:3.12}) follows from an application of the optional sampling theorem because $\{ m(W^\pi_{t \wedge \tau}) + A^\pi_{t \wedge \tau} \}_{t \ge 0}$ is a sub-martingale and $\sup_{t \ge 0} \E^{w, a} [m(W^\pi_{t \wedge\tau}) + A^\pi_{t \wedge \tau}] < \infty$; see Theorem 3.15, page 17 and Theorem 3.22, page 19 of \cite{kn:karat}.  By taking the infimum in (\ref{eq:3.12}) over all admissible investment strategies, we obtain
\begin{equation}\label{eq:ineq}
\inf_\pi \E^{w, a} \left[A^\pi_\tau + \frac{1}{\la} \, {\bf 1}_{ \left\{ Z^\pi_\tau \le -L \right\}} \right] \ge m(w) + a.
\end{equation}
Now, by (\ref{eq:newM-L}), the left-hand side of (\ref{eq:ineq}) equals $M_L(w, a)$; thus,
$M_L(w, a) \ge m(w) + a$.

If the individual investor follows a strategy $\pv$ that minimizes the right-hand side of (\ref{eq:HJB}), then (\ref{eq:3.9}) is satisfied with equality, and applying the dominated convergence theorem yields
\begin{equation}
\E^{w, a} [m(W^{\pv}_{t \wedge \tau}) + A^{\pv}_{t \wedge \tau}] = m(w) + A,
\end{equation}
which implies that $\{m(W^{\pv}_{t \wedge \tau}) + A^{\pv}_{t \wedge \tau} \}_{t \ge 0}$ is a martingale.  By following the same line of argument as above, we obtain
\begin{equation}
M_L(w, a) = m(w) + a,
\end{equation}
which demonstrates that (\ref{eq:verf}) holds and that $\pv$ in (\ref{eq:pi}) is an optimal investment strategy for wealth lying in $(-L, c/r)$.
\end{proof}

\subsection{A related free-boundary problem}\label{sec:FBP}

Next, we introduce a free-boundary problem (FBP) whose concave solution is the dual, via the Legendre transform, of $M_L(w, a) - a$.  In the first subsection, we solve the FBP, and in the subsection following, we show that the convex dual of the solution of the FBP, indeed, equals $M_L(w, a) - a$.

Consider the following free-boundary problem (FBP) on $[0, y_L]$:
\begin{equation}\label{FBP}
\begin{cases}
\la \, \hat m(y) =  (\la - r) \, y \, \hat m'(y) + \delta \, y^2 \, \hat m''(y) + c \, y + {\bf 1}_{\{y > y_0\}} \; \hbox{ on } \; C = (0, y_L), \\
\hat m(0) = 0 \; \hbox{ and }  \hat m'(y_0) = 0,\\
\hat m(y_L) = 1/\la - L \, y_L \; \hbox{ and } \hat m'(y_L) = - L,
\end{cases}
\end{equation}
in which $0 \le y_0 \le y_L$ are the free boundaries and $\delta = {1 \over 2} \left( {\mu - r \over \sigma} \right)^2$.

\subsubsection{Solving the FPB \eqref{FBP}}

To solve the FBP in \eqref{FBP}, we consider the problem on the two domains:  (1) $0 \le y \le y_0$, and (2) $y_0 < y \le y_L$.  After solving the FBP on each domain separately, we will impose value matching at $y = y_0$ to determine $y_0$.

First, consider $0 \le  y \le y_0$.  On this domain, $\hat m$ solves
\begin{equation}\label{FBP1}
\begin{cases}
\la \, \hat m(y) =  (\la - r) \, y \, \hat m'(y) + \delta \, y^2 \, \hat m''(y) + c \, y, \\
\hat m(0) = 0 \; \hbox{ and } \; \hat m'(y_0) = 0.
\end{cases}
\end{equation}
The general solution of the ODE in (\ref{FBP1}) is of the form
\begin{equation}
\hat m(y) = D_1 \, y^{B_1} + D_2 \, y^{B_2} + {c \over r} \, y,
\end{equation}
in which $D_1$ and $D_2$ are constants to be determined, and $B_1$ and $B_2$ are given by
\begin{equation}\label{eq:B1}
B_1 = {1 \over 2 \delta} \left[ (r - \la + \delta) + \sqrt{(r - \la + \delta)^2 + 4 \delta \la} \right] > 1,
\end{equation}
and
\begin{equation}\label{eq:B2}
B_2 = {1 \over 2 \delta} \left[ (r - \la + \delta) - \sqrt{(r - \la + \delta)^2 + 4 \delta \la} \right] < 0.
\end{equation}

The boundary condition at $y = 0$ implies that $D_2 = 0$.  The boundary condition at $y = y_0$ implies that
\begin{equation}
0 = \hat m'(y_0) = D_1 \, B_1 \, y_0^{B_1 - 1} + {c \over r},
\end{equation}
which gives us a relationship between $D_1$ and $y_0$, specifically,
\begin{equation}\label{eq:D1y0}
D_1 = - {c \over r B_1} \, y_0^{1 - B_1}.
\end{equation}
We have, therefore, shown that on $[0, y_0]$, $\hat m$ is given by
\begin{equation}\label{eq:hatm1}
\hat m(y) =  {c \over r} \, y \left[ 1 - {1 \over B_1} \left( {y \over y_0} \right)^{B_1 - 1} \right].
\end{equation}

Next, consider $y_0 < y \le y_L$.  On this domain, $\hat m$ solves
\begin{equation}\label{FBP2}
\begin{cases}
\la \, \hat m(y) =  (\la - r) \, y \, \hat m'(y) + \delta \, y^2 \, \hat m''(y) + c \, y + 1, \\
\hat m(y_L) = 1/\la - L \, y_L, \; \hat m'(y_L) = - L, \; \hbox{ and } \; \hat m'(y_0) = 0.
\end{cases}
\end{equation}
The general solution of the ODE in (\ref{FBP2}) is of the form
\begin{equation}\label{eq:ODE2}
\hat m(y) = d_1 \, y^{B_1} + d_2 \, y^{B_2} + {c \over r} \, y + {1 \over \la},
\end{equation}
in which $d_1$ and $d_2$ are constants to be determined, and $B_1$ and $B_2$ are given by (\ref{eq:B1}) and (\ref{eq:B2}), respectively.

The boundary conditions at $y = y_L$ imply that
\begin{equation}
{1 \over \la} - L y_L = \hat m(y_L) = d_1 \, y_L^{B_1} + d_2 \, y_L^{B_2} + {c \over r} \, y_L + {1 \over \la},
\end{equation}
and
\begin{equation}
-L = \hat m'(y_L) = d_1 \, B_1 \, y_L^{B_1 - 1} + d_2 \, B_2 \, y_L^{B_2 - 1} + {c \over r}.
\end{equation}
Solve these equations for $d_1$ and $d_2$ in terms of $y_L$ to obtain
\begin{equation}\label{eq:d1}
d_1 = - {1 - B_2 \over B_1 - B_2} \, y_L^{1- B_1} \left( {c \over r} + L \right),
\end{equation}
and
\begin{equation}\label{eq:d2}
d_2 = - {B_1 - 1 \over B_1 - B_2} \, y_L^{1- B_2} \left( {c \over r} + L \right).
\end{equation}
Substitute these expressions for $d_1$ and $d_2$ into (\ref{eq:ODE2}) to obtain, for $y_0 < y \le y_L$,
\begin{equation}\label{eq:hatm2}
\hat m(y) = y \left[ {c \over r} - {1 - B_2 \over B_1 - B_2} \, \left( {c \over r} + L \right) \, \left( {y \over y_L} \right)^{B_1-1} - {B_1 - 1 \over B_1 - B_2} \, \left( {c \over r} + L \right) \, \left( {y \over y_L} \right)^{B_2-1} \right] + {1 \over \la}.
\end{equation}
For later reference, we also have
\begin{equation}\label{eq:hatm-der}
\hat m'(y) =  {c \over r} - {B_1(1 - B_2) \over B_1 - B_2} \, \left( {c \over r} + L \right) \, \left( {y \over y_L} \right)^{B_1-1} - {(B_1 - 1) B_2 \over B_1 - B_2} \, \left( {c \over r} + L \right) \, \left( {y \over y_L} \right)^{B_2-1}.
\end{equation}

The boundary condition at $y = y_0$, namely $\hat m'(y_0) = 0$, gives us the following equation for the ratio $y_0/y_L \in (0, 1)$:
\begin{equation}\label{eq:y0yL}
{c \over c + rL} = {B_1 (1 - B_2) \over B_1 - B_2} \left( {y_0 \over y_L} \right)^{B_1 - 1} + {(B_1 - 1) B_2 \over B_1 - B_2} \left( {y_0 \over y_L} \right)^{B_2 - 1}.
\end{equation} 
To see that (\ref{eq:y0yL}) has a unique solution $y_0/y_L$ in $(0, 1)$, note that (a) as $y_0/y_L$ approaches 0, the right-hand side of this equation approaches $- \infty$; (b) when $y_0/y_L = 1$, the right-hand side equals $1 > c/(c + rL)$; and (c) the right-hand side is strictly increasing with respect to $y_0/y_L$.

Finally, given $y_0/y_L$, we can obtain $y_0$ from $\hat m(y_0-) = \hat m(y_0+)$:
\begin{equation}
{c \over r} \, {B_1 - 1 \over B_1} \, y_0  = y_0 \left[ {c \over r} - {1 - B_2 \over B_1 - B_2} \, \left( {c \over r} + L \right) \, \left( {y_0 \over y_L} \right)^{B_1-1} - {B_1 - 1 \over B_1 - B_2} \, \left( {c \over r} + L \right) \, \left( {y_0 \over y_L} \right)^{B_2-1} \right] + {1 \over \la},
\end{equation}
or equivalently,
\begin{equation}\label{eq:y0}
y_0 = {1 \over \la} \left[ - {1 \over B_1} \, {c \over r} + {1 - B_2 \over B_1 - B_2} \, \left( {c \over r} + L \right) \, \left( {y_0 \over y_L} \right)^{B_1-1} + {B_1 - 1 \over B_1 - B_2} \, \left( {c \over r} + L \right) \, \left( {y_0 \over y_L} \right)^{B_2-1}  \right]^{-1}.
\end{equation}

Thus, we have proved the following proposition:

\begin{proposition}\label{prop:FBP}
The solution $\hat m$ of the free-boundary problem in $\eqref{FBP}$ is given by
\begin{equation}
\hat m(y) =
\begin{cases}
{c \over r} \, y \left[ 1 - {1 \over B_1} \left( {y \over y_0} \right)^{B_1 - 1} \right], & 0 \le y \le y_0; \\
y \left[ {c \over r} - {1 - B_2 \over B_1 - B_2} \, \left( {c \over r} + L \right) \, \left( {y \over y_L} \right)^{B_1-1} - {B_1 - 1 \over B_1 - B_2} \, \left( {c \over r} + L \right) \, \left( {y \over y_L} \right)^{B_2-1} \right] + {1 \over \la}, & y_0 < y \le y_L.
\end{cases}
\end{equation}
The free boundaries $0 \le y_0 \le y_L$ are given by solving $\eqref{eq:y0yL}$ for $y_0/y_L \in (0, 1)$, then by obtaining $y_0$ from \eqref{eq:y0} and $y_L$ from the relationship $y_L = y_0/(y_0/y_L)$.
\end{proposition}

Note that the solution $\hat m$ of the FBP is bounded above by $u_L$ and that the two curves are tangent at $y = 0$ and $y = y_L$, in which $u_L$ is given by
\begin{equation}\label{eq:u-L}
u_L(y) = \min \left( {c \over r} \, y, \; {1 \over \la} - L \, y \right).
\end{equation}
We will use this observation when we consider the limit as $L \to \infty$.

The following lemma shows that $\hat m$ is concave, a property we use in the next section.

\begin{lemma}\label{lem:concave}
The solution $\hat m$ of the FBP in $\eqref{FBP}$ is concave on $[0, y_L]$.
\end{lemma}

\begin{proof}
On $[0, y_0)$, $\hat m$ is given by \eqref{eq:hatm1}, with first derivative $\hat m'(y) = (c/r) \left[ 1 - (y/y_0)^{B_1 - 1} \right] > 0$ and second derivative $\hat m''(y) = -(c/r) (B_1 - 1) y^{B_1 - 2}/y_0^{B_1 - 1} < 0$.  On $(y_0, y_L]$, $\hat m$ is given by \eqref{eq:hatm2}.  One can show that the first derivative of $\hat m$ is negative by using the fact that the right-hand side of \eqref{eq:y0yL} is strictly increasing with respect to $y_0/y_L$.  Also, the second derivative
\begin{equation}\label{eq:hatm-der2}
\hat m''(y) = - {1 \over y_L} {(B_1 - 1) (1 - B_2) \over B_1 - B_2}  \left( {c \over r} + L \right) \left[ B_1 \left( {y \over y_L} \right)^{B_1 - 2} - B_2 \left( {y \over y_L} \right)^{B_2 - 2} \right] < 0.
\end{equation}
Finally, $\hat m'(y_0) = 0$; thus, we have shown that $\hat m$ is concave on $[0, y_L]$.
\end{proof}

In the next section, we show that $\hat m$ is intimately connected with the minimum occupation time $M_L$.

\subsubsection{Relation between the FBP (\ref{FBP}) and the minimum occupation time}

In this section, we show that the Legendre transform (see, for example, \cite{kn:karat2}) of the solution of the FBP (\ref{FBP}) is, in fact, the minimum occupation time $M_L(w, 0) = M_L(w, a)) - a$ for $w \in (-L, c/r)$.

To this end, first note that $\hat m$ is concave on $[0, y_L]$; thus, we can define its convex dual via the Legendre transform.  Indeed, for $w \ge -L$, define
\begin{equation}\label{Leg}
\ML(w) = \max_{y \ge 0} \left( \hat m(y) - wy \right).
\end{equation}
We will show that $\ML$ satisfies the conditions of Theorem \ref{thm:verf} and thereby show that $m(w) = M_L(w, a) - a = M_L(w, 0)$ for $w \in (-L, c/r)$.

For $w \in (-L, c/r)$, the critical value $y^*$ solves the equation $\hat m'(y) - w = 0$; thus, $y^* = I(w)$, in which $I$ is the inverse of $\hat m'$.  Note that $y_0 = I(0)$ and that $y > y_0$ corresponds to $w < 0$.  It follows that
\begin{equation}\label{3.2}
\ML(w) = \hat m \left( I(w) \right) - w I(w).
\end{equation}
Expression (\ref{3.2}) implies that
\begin{equation}\label{3.3}
\ML'(w) = \hat m'\left( I(w) \right) I'(w) - I(w) - w I'(w) = w I'(w) - I(w) - w I'(w) = - I(w).
\end{equation}
\noi Thus, $y^* = I(w) = - \ML'(w)$. Note that from (\ref{3.3}), we have
\begin{equation}\label{3.5}
\ML''(w) = -I'(w) = - {1 \over \hat m'' \left( I(w) \right)}.
\end{equation}
It is also straightforward to show the following relationship, which we use in the next section when we investigate properties of $M_L$: 
\begin{equation}\label{eq:third-der}
\ML'''(w) = {\hat m''' \left( I(w) \right) \over \left( \hat m'' \left( I(w) \right) \right)^3}.
\end{equation}

\begin{thm}\label{thm:dual}
Let $\ML$ be given by $(\ref{Leg}),$ in which $\hat m$ is given by Proposition $\ref{prop:FBP}$.  Then, the minimum occupation time $M_L$ defined in $(\ref{eq:M-L})$ is related to $\ML$ by $M_L(w, a) = \ML(w) + a$ for $(w, a) \in (-L, c/r) \times \R^+$.
\end{thm}

\begin{proof}
We first find the boundary-value problem (BVP) that $\ML$ solves given that $\hat m$ solves (\ref{FBP}). In the ODE for $\hat m$ in (\ref{FBP}), let $y = I(w) = -\ML'(w)$ to obtain
\begin{equation}
\la \, \hat m\left( I(w) \right) =  (\la - r) I(w) \hat m'\left( I(w) \right) + \delta I^2(w) \hat m''\left( I(w) \right) + c I(w) + {\bf 1}_{\{w < 0\}}.
\end{equation}
\noi Rewrite this equation in terms of $\ML$ to get
\begin{equation}\label{fbp1}
\la \left( \ML(w) - w \ML'(w) \right) = - (\la - r) \ML'(w) w + \delta {\left( \ML''(w) \right)^2 \over - \ML''(w)} - c \ML'(w) + {\bf 1}_{\{w < 0\}},
\end{equation}
\noi or equivalently,
\begin{equation}\label{3.7}
\la \, \ML(w) = {\bf 1}_{\{w < 0\}} +  (rw - c) \ML'(w) - \delta { \left( \ML''(w) \right)^2 \over \ML''(w)}.
\end{equation}
Note that because $\ML$ is convex, $- \delta \, \left( \ML''(w) \right)^2/ \ML''(w)$ can be written as \hfill \break $\min_\pi \left[ (\mu - r) \pi \ML'(w) + {1 \over 2} \sigma^2 \pi^2 \ML''(w) \right]$, which is similar to the minimization term in (\ref{eq:HJB}).

Next, consider the boundary and terminal conditions in (\ref{FBP}).  First, the boundary condition at $y = 0$,
namely $\hat m(0) = 0$, and $\hat m'(0) = c/r$ from \eqref{eq:hatm1} imply that the corresponding dual value of $w$ is $c/r$ and that
\begin{equation}\label{fbp3}
\ML(c/r) = 0.
\end{equation}
Similarly, the boundary conditions at $y = y_L$, namely $\hat m(y_L) = 1/\la - L y_L$ and $\hat m(y_L) = -L$, imply that the corresponding dual value of $w$ is $-L$ and that
\begin{equation}\label{fbp4}
\ML(-L) = {1 \over \la}.
\end{equation}

Because $\ML$ satisfies the conditions given in Theorem \ref{thm:verf}, we have shown that the minimum occupation time $M_L$  for $(w, a) \in (-L, c/r) \times \R^+$ is given by $M_L(w, a) = \ML(w) + a$.  
\end{proof}

Thus, we can obtain the minimum expected occupation time $M_L$ as described in the following corollary to Theorem \ref{thm:dual}.

\begin{cor}\label{cor:dual}
For $(w, a) \in (-L, c/r) \times \R^+,$ the minimum expected occupation time $M_L$ is given by $M_L(w, a) = m(w) + a,$ in which $m$ is given by
\begin{equation}\label{eq:M-L-via-dual}
m(w) = 
\begin{cases}
\beta_L \left( {c \over r} - w \right)^p, & 0 \le w < c/r; \\
y \left[ \left( {c \over r} - w \right) - {1 - B_2 \over B_1 - B_2} \, \left( {c \over r} + L \right) \, \left( {y \over y_L} \right)^{B_1-1} - {B_1 - 1 \over B_1 - B_2} \, \left( {c \over r} + L \right) \, \left( {y \over y_L} \right)^{B_2-1} \right] + {1 \over \la}, & -L < w < 0.
\end{cases}
\end{equation}
In the first expression of $\eqref{eq:M-L-via-dual},$ the constant $p$ equals  $B_1/(B_1 - 1) > 1$.  In the second expression of $\eqref{eq:M-L-via-dual},$ $y$ is the solution of $w = \hat m'(y),$ in which $\hat m'$ is given by $\eqref{eq:hatm-der}$.  Finally, $\beta_L$ is given by continuity of $m$ and $m'$ at $w = 0$ and solves
\begin{equation}\label{eq:y0betaL}
y_0 = \beta_L \, p \, \left( {c \over r} \right)^{p-1}.
\end{equation}
Here, $y_0$ and $y_L$ are as in Proposition~\ref{prop:FBP}.
\end{cor}

\section{Properties of the Minimum Occupation Time and Optimal Investment Strategy}\label{sec:prop}

In this section, we address the following questions.

\begin{enumerate}
\item{} How does the optimal investment strategy $\pi^*_L(w)$ given in \eqref{eq:pi} compare with the optimal investment strategy when minimizing the probability of lifetime ruin, as determined in \cite{by3}, namely ${\mu - r \over \sigma^2} \, {c/r - w \over p - 1}$, in which $p = B_1/(B_1 - 1)$?

\item{} How does $\pi^*_L(w)$ vary with respect to $w$?

\item{} How do $\pi^*_L$ and $M_L$ change as $L \to \infty$?
\end{enumerate}

\subsection{Optimal investment strategies}

Let $\pi^*_0$ denote the optimal investment strategy when minimizing the probability of lifetime ruin.  Recall that this is the optimal investment strategy corresponding to any ruin level; that is, it is independent of the ruin level.  Then, we have the following proposition.

\begin{proposition}\label{prop:pi}
$\pi^*_L(w) = \pi^*_0(w)$ for $0 <  w < c/r,$ and $\pi^*_L(w) > \pi^*_0(w)$ for $-L < w < 0$.
\end{proposition}

\begin{proof}
For $0 < w < c/r$, Corollary \ref{cor:dual} gives us that $m(w) = \beta_L (c/r - w)^p$.  Then, from the expression in \eqref{eq:pi}, we obtain that $\pi^*_L(w) = {\mu - r \over \sigma^2} \, {c/r - w \over p - 1} = \pi^*_0(w)$.

For $-L < w < 0$,
\begin{equation}\label{eq:lev}
\pi^*_L(w) = - {\mu - r \over \sigma^2} \, {m'(w) \over m''(w)} > {\mu - r \over \sigma^2} \, {c/r - w \over p - 1}
\end{equation}
 if and only if
 \begin{equation}\label{eq:lev2}
 -y \, \hat m''(y) > (B_1 - 1) (c/r - \hat m'(y)),
 \end{equation}
 for $y_0 < y < y_L$, by substituting $w = \hat m'(y)$ into \eqref{eq:lev}.  By substituting for $\hat m'$ and $\hat m''$ from \eqref{eq:hatm-der} and \eqref{eq:hatm-der2}, respectively, into \eqref{eq:lev2} and simplifying the result, it becomes
\begin{equation}\label{eq:leverage}
-B_2 (1 - B_2) (y/y_L)^{B_2 - 1} > B_2 (B_1 - 1) (y/y_L)^{B_2 - 1}.
\end{equation}
The inequality in \eqref{eq:leverage} is true because the left-hand side is positive, while the right-hand side is negative.
\end{proof}

Proposition \ref{prop:pi} implies that if one seeks to minimize expected occupation time, then leveraging is {\it worse} for negative wealth than when minimizing the probability of lifetime ruin.

\subsection{$\pi^*_L$ as a function of $w$}

When minimizing the probability of lifetime ruin, the optimal investment strategy is a linear, decreasing function of $w$.  From the first part of Proposition \ref{prop:pi}, we know that the same is true for $\pi^*_L(w)$ for $0 < w < c/r$.  Therefore, we wish to see how $\pi^*_L(w)$ varies with respect to $w$ for $-L < w < 0$.

\begin{proposition}\label{prop:pi-mono}
$\pi^*_L(w)$ decreases on $(-L, 0)$ if and only if $B_1 (B_1 - 1) \left( {y_0 \over y_L} \right)^{B_1 - B_2} > - B_2 (1 - B_2$.  If $r < \la,$ then $\pi^*_L(w)$ increases on $(-L, 0)$.
\end{proposition}

\begin{proof}
>From the expression in \eqref{eq:pi}, we know that $\pi^*_L(w)$ decreases with respect to $w$ if and only if $-m'(w)/m''(w)$ decreases with respect to $w$, or equivalently by differentiating, if and only if
\begin{equation}\label{eq:pi-decr}
m'(w) m'''(w) < \left( m''(w) \right)^2.
\end{equation}
By setting $w = \hat m'(y)$, or equivalently $y = I(w) = -m'(w)$, and by using \eqref{3.5} and \eqref{eq:third-der}, inequality \eqref{eq:pi-decr} becomes
\begin{equation}\label{eq:pi-decr2}
y \, \hat m'''(y) < - \hat m''(y).
\end{equation}

>From \eqref{eq:hatm-der2}, we obtain that for $y_0 < y < y_L$,
\begin{equation}\label{eq:hatm-der3}
\hat m'''(y) = - {1 \over y_L^2} {(B_1 - 1) (1 - B_2) \over B_1 - B_2}  \left( {c \over r} + L \right) \left[ B_1(B_1 - 2) \left( {y \over y_L} \right)^{B_1 - 3} - B_2 (B_2 - 2) \left( {y \over y_L} \right)^{B_2 - 3} \right].
\end{equation}
By substituting for $\hat m''$ and $\hat m'''$ from \eqref{eq:hatm-der2} and \eqref{eq:hatm-der3}, respectively, into \eqref{eq:pi-decr2} and by simplifying the result, the inequality becomes
\begin{equation}\label{eq:pi-decr3}
B_1 (B_1 - 1) \left( {y \over y_L} \right)^{B_1 - B_2} + B_2 (1 - B_2) > 0.
\end{equation}
The inequality in \eqref{eq:pi-decr3} holds for all $y_0 < y < y_L$ if and only if it holds at $y = y_0$, and this completes the proof of the first part of the proposition.

By a similar line of reasoning, we learn that $\pi^*_L(w)$ increases on $(-L, 0)$ if and only if 
\begin{equation}\label{eq:pi-incr}
B_1 (B_1 - 1) \left( {y \over y_L} \right)^{B_1 - B_2} + B_2 (1 - B_2) < 0,
\end{equation}
for all $y_0 < y < y_L$, which occurs if and only if \eqref{eq:pi-incr} holds at $y = y_L$.  It is straightforward to show that $B_1 (B_1 - 1) < - B_2 (1 - B_2)$ if and only if $r < \la$.
\end{proof}

Proposition \ref{prop:pi-mono} shows us that if the rate of return on the riskless asset is low enough, then the individual will borrow more money to invest in the risky asset as wealth gets closer to zero.  In other words, because the borrowing rate is low, the individual can take on more debt (because debt is relatively cheap) in order to get wealth above zero.

\subsection{Varying $L$}

In this section, we examine how $M_L$ and $\pi^*_L$ change as $L$ increases.

\begin{proposition}\label{prop:pi-incr}
If $L_1 \le L_2,$ then $\pi^*_{L_1}(w) \le \pi^*_{L_2}(w)$ for $w \in (-L_1, c/r)$.
\end{proposition}

\begin{proof}
Because $\pi^*_L(w)$ is independent of $L$ for $w \in [0, c/r)$, it is enough to show that $\partial{\pi^*_L(w)} / \partial L > 0$ for $w \in (-L, 0)$.  Showing this inequality is equivalent to showing that $-{\partial{} \over \partial L} y \, \hat m''(y) > 0$ for $y \in (y_0, y_L)$.  From \eqref{eq:hatm-der2}, this inequality is equivalent to
\begin{equation}
{\partial {} \over \partial{L}} \left( {c \over r} + L \right) \left[ B_1\left( {y \over y_L} \right)^{B_1 - 1} - B_2 \left( {y \over y_L} \right)^{B_2 - 1} \right] > 0, \qquad y_0 < y < y_L,
\end{equation}
or
\begin{align}\label{eq:ineq-pi}
& \left( {c \over r} + L \right) \left[ B_1 (B_1 - 1) \left( {y \over y_L} \right)^{B_1 - 2} + B_2 (1 - B_2) \left( {y \over y_L} \right)^{B_2 - 2} \right] {\partial{(y/y_L)} \over \partial{L}} \notag \\
& \qquad + \left[ B_1\left( {y \over y_L} \right)^{B_1 - 1} - B_2 \left( {y \over y_L} \right)^{B_2 - 1} \right]  > 0, \qquad y_0 < y < y_L.
\end{align}
>From $w = \hat m'(y)$ for $-L < w < 0$, we compute ${\partial{(y/y_L)} \over \partial{L}}$ by differentiating with respect to $L$ and obtain
\begin{align}\label{eq:der-yyL}
& \left( {c \over r} + L \right) (B_1 - 1)(1 - B_2) \left[ B_1\left( {y \over y_L} \right)^{B_1 - 2} - B_2 \left( {y \over y_L} \right)^{B_2 - 2} \right] {\partial{(y/y_L)} \over \partial{L}} \notag \\
& = - \left[ B_1(1 - B_2) \left( {y \over y_L} \right)^{B_1 - 1} + (B_1 - 1) B_2 \left( {y \over y_L} \right)^{B_2 - 1} \right].
\end{align}
By substituting ${\partial{(y/y_L)} \over \partial{L}}$ from \eqref{eq:der-yyL} into the inequality in \eqref{eq:ineq-pi}, we obtain the following inequality
\begin{align}
& - \left[ B_1(B_1 - 1) \left( {y \over y_L} \right)^{B_1 - B_2}  +  B_2(1 - B_2) \right] \, \left[ B_1(1 - B_2) \left( {y \over y_L} \right)^{B_1 - B_2}  + (B_1 - 1)B_2 \right]  \notag \\
& \qquad + (B_1 - 1)(1 - B_2) \left[ B_1 \left( {y \over y_L} \right)^{B_1 - B_2} - B_2 \right]^2 > 0, \qquad y_0 < y < y_L,
\end{align}
which after simplifying, is clearly true.
\end{proof}

Proposition \ref{prop:pi-incr} states that $\pi^*_L$ increases as $L$ increases.  That gives rise to the question as to how far $\pi^*_L$ increases?  That is, what is its limiting value as $L \to \infty$?  The next proposition answers that question.

\begin{proposition}\label{prop:pi-limit}
As $L \to \infty,$ $\pi^*_L(w) \to \infty$ linearly with respect to $L,$ for all $w < 0$.  
\end{proposition}

\begin{proof}
Fix $w < 0$; then, for $L < w$, the dual value $y$ corresponding to $w$ solves $w = \hat m'(y)$ from \eqref{eq:hatm-der}, or equivalently,
\begin{equation}\label{eq:y-limit}
{c - r w \over c + r L} =  {B_1(1 - B_2) \over B_1 - B_2} \, \left( {y \over y_L} \right)^{B_1-1} + {(B_1 - 1) B_2 \over B_1 - B_2} \, \left( {y \over y_L} \right)^{B_2-1}.
\end{equation}
As $L$ gets large, the left-hand side of \eqref{eq:y-limit} goes to $0$, while the solution of $y/y_L$ goes to $z$, which solves
\begin{equation}
z^{B_1 - B_2} = - {B_1 - 1 \over B_1} \, {B_2 \over 1 - B_2},
\end{equation}
independent of $w$.

We also have
\begin{align}\label{eq:pi-limit}
\lim_{L \to \infty} \pi^*_L(w) & = - \lim_{L \to \infty}  {\mu - r \over \sigma^2} \, {m'(w) \over m''(w)} \notag  \\
& \propto - \lim_{L \to \infty} y \, \hat m''(y), \qquad \qquad \qquad \hbox{for } y = -m'(w), \notag \\
& \propto \lim_{L \to \infty} \left(   B_1 \left( {y \over y_L} \right)^{B_1 - 1} - B_2 \left( {y \over y_L} \right)^{B_2 - 1}  \right) \left( {c \over r} + L \right).
\end{align}
Here, $a \propto b$ means that there exists a {\it positive} quantity $k$ independent of $L$ such that $a = kb$.  It follows that
\begin{equation}\label{eq:pi-limit2}
\lim_{L \to \infty} \pi^*_L(w) \propto \left(   B_1 z^{B_1 - 1} - B_2 z^{B_2 - 1}  \right) \, \lim_{L \to \infty}  \left( {c \over r} + L \right) = \infty.
\end{equation}
linearly with respect to $L$.
\end{proof}

\begin{proposition}\label{prop:ML-decr}
If $L_1 < L_2,$ then $M_{L_1}(w, 0) > M_{L_2}(w, 0)$ for $w \in (-L_1, \infty)$.
\end{proposition}

\begin{proof}
See the Appendix.
\end{proof}

Thanks to Proposition \ref{prop:ML-decr} we can show that the sequence of function defined in \eqref{eq:M-L} converges to the value function of minimizing lifetime occupation time. 
\begin{proposition}
For any $w \in \R$ and $a\geq 0$, we have that
\begin{equation}
\lim_{L \rightarrow \infty}M_{L}(w,a)=\inf_{\pi}E^{w,a}[A_{\tau_d}].
\end{equation}
\end{proposition}
\begin{proof}
It follows from Corollary~\ref{cor:dual} that $M_L(w,a)=a+M_L(w,0)$. Thanks to Proposition \ref{prop:ML-decr} we know that if $L_1 < L_2,$ then $M_{L_1}(w, a) > M_{L_2}(w, a)$ for $w \in (-L_1, \infty)$. However, when $w \leq -L_1$ we have that $M_{L_1}(w,a)=M_{L_2}(w,a)$, which follows directly from \eqref{eq:M-L}. As a result
\[
\lim_{L \rightarrow \infty}M_{L}(w,a)=\inf_{L>0}M_{L}(w,a). 
\]
But then,
\[
\begin{split}
\lim_{L \rightarrow \infty}M_{L}(w,a)&=\inf_{L>0}\inf_\pi \E^{w, a} \left[ A_{\tau_d} \, {\bf 1}_{\{Z_{\tau_d} > -L \}} + \left( A_{\tau_L} + \frac{1}{\la} \right)  {\bf 1}_{\{Z_{\tau_d} \le -L \}} \right]
\\&=\inf_\pi \inf_{L>0} \E^{w, a} \left[ A_{\tau_d} \, {\bf 1}_{\{Z_{\tau_d} > -L \}} + \left( A_{\tau_L} + \frac{1}{\la} \right)  {\bf 1}_{\{Z_{\tau_d} \le -L \}} \right]
\\&=\inf_{\pi} \E^{w, a} \left[ A_{\tau_d} \right].
\end{split}
\]
\end{proof}

Proposition \ref{prop:ML-decr} gives rise to the question as to how far $M_L$ decreases as $L$ increases without bound?  The next proposition answers that question.

\begin{proposition}\label{prop:M-L-limit}
For all $w \in \R,$ $\lim_{L \to \infty} M_L(w, 0) =\inf_{\pi} \E^{w, a} \left[ A_{\tau_d} \right]= 0$.
\end{proposition}

\begin{proof}
First note that from \eqref{eq:u-L} and the observation accompanying it, we have that the dual variable $y \to 0$ as $L \to \infty$ for all $w < 0$.  Recall that the dual variable $y = -m'(w)$; thus, as $L \to \infty$, it follows that $m'(w) \to 0$ for all $w < 0$.  In particular, $y_0 \to 0$, which implies that $\beta_L \to 0$ in \eqref{eq:y0betaL}.  Thus, the conclusion of the proposition follows.
\end{proof}

The minimum expected occupation time goes to 0 (and trivially convex and non-decreasing on all of $\R$) because the corresponding investment strategy grows linearly with $L$ (see Proposition~\ref{prop:pi-limit}).  In other words, as soon as wealth gets negative, the investment strategy becomes infinitely large and thereby leverages the wealth back into positive territory with probability 1.

\begin{remark}
When the wealth is negative the individual is borrowing in order to fund her consumption. It is natural to impose a higher borrowing rate, say $b>r$, for individuals who have negative wealth. However, having a higher borrowing rate for negative wealth would not change qualitative results of this paper. In particular, the Legenre transform of value function $\hat{m}$ on $(-\infty,0)$ can be calculated as in (\ref{FBP2})-\eqref{eq:y0} just by replacing $r$ with $b$, recalling that $B_1$ and $B_2$ appearing in these expressions are also functions of $r$. Now, it can be easily checked that the statements we  proved in this section for the optimal investment strategy $\pi(w)=-(\mu-r)/\sigma^2 y \hat{m}''(y)$, in which $w=\hat{m}'(y)$, and the value function remain exactly the same.

One might also try to impose a restriction on trading by setting $\pi_t=0$ when $W_t \leq 0$. In this case $\E^{w,a}[A_{\tau_d}]=1/\lambda$, for $w \leq 0$, since the individual has no chance to recover from bankruptcy. As a result, with this restriction on trading, the problem of minimizing the occupation time is equivalent to the problem of minimizing the probability of lifetime ruin; see \cite{by3}.

\end{remark}

\section{Extension}\label{sec:exten}

\subsection{More General Assumptions on $\tau_d$}

Note that we can write \eqref{eq:M-L} as 
\begin{equation}
M_L(w,a)=\inf_{\pi}\E\left[\int_0^{\tau_L}\lambda e^{-\lambda s}A_s ds+e^{-\lambda \tau_L}\left(A_{\tau_L}+{1 \over \lambda}\right) \right],
\end{equation}
and thanks to \eqref{thm:dual} we know that $M$ solves 
\begin{equation}
\lambda (M_L-a)=1_{\{w<0\}}+(rw-c)M_L'+\min_{\pi}\left[(\mu-r)\pi M_L'+{1 \over 2} \sigma^2 \pi^2 M_L''\right].
\end{equation}
When $\tau_d$ is the $n$th jump time of a Poisson process (i.e. $\tau_d$ has Erlang distribution) we expect that $M_L(w,a)=M^{(n)}(w,a)$ in which $M^{(0)}=a$ and
\begin{equation}
M^{(k)}(w,a)=\inf_{\pi}\E\left[\int_0^{\tau_L}\lambda e^{-\lambda s}M^{(m-1)}(W_s,A_s)ds+e^{-\lambda \tau_L}\left(A_{\tau_L}+{1 \over \lambda}\right) \right],
\end{equation}
for $k \in \{1, \cdots,n\}$. Moreover, we expect that for each $m \in (-L,c/r)$, $M^{(k)}$ will be a classical solution of \begin{equation}
\lambda (M^{(k)}-M^{(m-1)})=1_{\{w<0\}}+(rw-c)(M^{(k)})'+\min_{\pi}\left[(\mu-r)\pi (M^{(k)})'+{1 \over 2} \sigma^2 \pi^2 (M^{(k)})''\right],
\end{equation}
and that $M^{(k)}(w,a)=a+m^{(k)}(w)$, in which $m^{(k)}$ satisfies conditions (i)-(iv) in Theorem~\ref{thm:verf}. In our future work, we will make these statements more rigorous and analyze the effect of of increasing $n$ on the optimal investment of the retiree. 

A generalization of Erlang distribution for $\tau_d$ would be to consider the time of absorption of a continuous time discrete state space Markov chain with one absorbing state. In this case $\tau_d$ is said to have a phase-type distribution, which is dense in the set of all positive-valued distributions, that is, it can be used to approximate any positive valued distribution; see e.g. \cite{neuts}. 
We will next describe how one can solve for $M_L$ with this assumption on $\tau_d$.
Let $a_{ij}$ be the rate at which this Markov
chain jumps from state $i$ to state $j$. Then the action of
infinitesimall generator $\mathcal{A}$ of this Markov chain on a
test function $f:\{1,2,...,d\} \rightarrow \R$ is given by
\begin{equation}
\mathcal{A}f(i)=\sum_{j \neq i}a_{ij}[f(j)-f(i)], \quad i\in
\{1,2,...,d\}.
\end{equation}
If we assume that there is a family of functions $m^{(i)}(w)$, $i \in\{1,\cdots,d\}$ solving 
the coupled non-linear ordinary differential equations
\begin{equation}
\sum_{j \neq i}a_{ij} (m^{(i)}-m^{(j)})=1_{\{w<0\}}+(rw-c)(m^{(i)})'+\min_{\pi}\left[(\mu-r)\pi (m^{(i)})'+{1 \over 2} \sigma^2 \pi^2 (M^{(i)})''\right], 
\end{equation}
on $w \in (-L,c/r)$ and satisfy the conditions (i)-(iv) in Theorem~\ref{thm:verf},
then we can show that $M_L(w,a)=m^{(k)}+a$, in which $k \in \{1,2,...,d\}$ is the absorbing state. The verification of this assumption and detailed analysis of this problem will be the topic of our future work.

In general, when $\P(\tau_d>t)=\exp(-\int_{0}^{t}\lambda(s)ds)$, for some positive function $\lambda$,
then $M_L$ will also be a function of the current time and is expected to satisfy a non-linear partial differential equation
\begin{equation}
\lambda(t) (M_L-a)={\partial M_L \over \partial t}+1_{\{w<0\}}+(rw-c){\partial M_L \over \partial w}+\min_{\pi}\left[(\mu-r)\pi {\partial M_L \over \partial w}+{1 \over 2} \sigma^2 \pi^2 {\partial^2 M_L \over \partial^2 w}\right],
\end{equation}
on $w \in (-L,c/r)$ and $t \geq 0$ and satisfy the boundary conditions dictated by the conditions in  Theorem~\ref{thm:verf}. Legendre transform linearizes this partial differential equation at the expense of introducing free boundaries. This generalization will also be analyzed along with the two cases we mentioned above in our future work. 

Guided by the results on minimizing probability of ruin we expect that the value function $M_L(w,a,0)$ will not change much if the expected value of $\tau_d$ is kept the same; however the distributional assumptions on $\tau_d$ is expected to have a significant impact on the optimal investment strategy (see Figure 6.3 of \cite{mmy06}). However, approximate investment strategies based on the one we obtained in this paper might give nearly optimal strategies: An investor might estimate her (constant) hazard rate $\lambda$ each year and rebalance her portfolio using $\pi^*$ we prescribed in this paper. Then she waits without taking any further action until the next year and rebalances her portfolio given her new estimate for the hazard rate. In the problem of minimizing probability of ruin this discrete rebalancing approximation scheme was shown to be nearly optimal by \cite{my06}. A similar analysis needs to be undertaken for the problem of minimizing the occupation time.

\subsection{Minimizing the Expected Area Below Zero }
A natural extension to the problem of minimizing expected occupation time is to replace $A$ in the definition of $M_L$ in \eqref{eq:M-L} with $A^f$ defined by
\begin{equation}\label{eq:Af}
A^f_t = A^f_0 + \int_0^t f(W_s) \, ds, \quad A^f_0 = a^f \ge 0,
\end{equation}
in which $f$ is some non-decreasing function of wealth, and to replace the penalty $1/\la$ in \eqref{eq:M-L} with $f(-L)/\la$.  Specifically, the value function $M^f_L$ is defined by
\begin{equation}\label{eq:Mf-L}
M^f_L(w, a^f) = \inf_\pi \E^{w, a^f} \left[ A^f_{\tau_d} \, {\bf 1}_{\{Z_{\tau_d} > -L \}} + \left( A^f_{\tau_L} + \frac{f(-L)}{\la} \right)  {\bf 1}_{\{Z_{\tau_d} \le -L \}} \right].
\end{equation}

In \eqref{eq:M-L}, the function $f$ is given by $f(w) = {\bf 1}_{\{w < 0 \}}$.  For another example, if the goal were to minimize the expected area between the negative part of the trajectory of wealth and the $w = 0$ horizontal line, then $f$ would be given by $f(w) = -w \, {\bf 1}_{\{w < 0 \}}$.

As in Theorem \ref{thm:verf}, we expect that $M^f_L$ to be given by $M^f_L(w, a^f) = m^f(w) + a^f$, in which $m^f$ satisfies the conditions listed in that theorem with appropriate changes.  Specifically, $m^f(w) = f(-L)/\la$ for all $w \le -L$ and $m^f$ solves the following HJB equation on $(-L, c/r)$:
\begin{equation}\label{eq:HJBf}
\la m^f(w) = f(w) + (rw - c) (m^f)'(w) + \min_{\pi} \left[(\mu-r)\pi (m^f)'(w) + \frac{1}{2}\sigma^2 \pi^{2} (m^f)''(w) \right].
\end{equation}

Now, in order that the corresponding FBP be linear as in Section \ref{sec:FBP}, the function $f$ must be piecewise linear.  Moreover, to be able to get closed-form solutions for the solutions of the ODEs in the FBP, $f$ must be piecewise constant.  By approximating a given function via a step function, one can approximate the value function $M^f_L$.

\setcounter{section}{0}%
\renewcommand{\thesection}{\Alph{section}}%
\section{Appendix: Proof of Proposition~\ref{prop:ML-decr}}

From the relationship between $\beta_L$ and $y_0$ in Corollary \ref{cor:dual}, we see that it is enough to show that $\partial{M_L(w, 0)}/\partial{L} < 0$ for $w \in (-L, 0)$.  To that end, recall that
\begin{equation}\label{eq:ML-def}
M_L(w, 0) = y \left( {c \over r} - w \right) - {1 - B_2 \over B_1 - B_2} \left( {c \over r} + L \right) {y^{B_1} \over y_L^{B_1 - 1}} - {B_1 - 1 \over B_1 - B_2} \left( {c \over r} + L \right) {y^{B_2} \over y^{B_2 -1}} + {1 \over \la},
\end{equation}
in which $w = \hat m'(y)$, or equivalently,
\begin{equation}\label{eq:w-hatm-der}
{c \over r} - w = \left( {c \over r} + L \right) \left[ {B_1 (1 - B_2) \over B_1 - B_2} \left( {y \over y_L} \right)^{B_1 - 1} + {(B_1 - 1) B_2 \over B_1 - B_2} \left( {y \over y_L} \right)^{B_2 - 1} \right].
\end{equation}
Differentiate \eqref{eq:ML-def} with respect to $L$ to obtain
\begin{align}\label{eq:ML-der}
{\partial{M_L(w, 0)} \over \partial{L}} &= \left( {c \over r} - w \right) {\partial{y} \over \partial{L}} - \left[ {1 - B_2 \over B_1 - B_2} \, {y^{B_1} \over y_L^{B_1 - 1}} + {B_1 - 1 \over B_1 - B_2} \, {y^{B_2} \over y^{B_2 -1}} \right] \notag \\
& \quad -  \left[ {B_1 (1 - B_2) \over B_1 - B_2} \left( {y \over y_L} \right)^{B_1 - 1} + {(B_1 - 1) B_2 \over B_1 - B_2} \left( {y \over y_L} \right)^{B_2 - 1} \right] \left( {c \over r} + L \right) {\partial{y} \over \partial{L}} \notag \\
& \quad +  \left[ {(B_1 - 1)(1 - B_2) \over B_1 - B_2} \left( {y \over y_L} \right)^{B_1} - {(B_1 - 1)(1 - B_2) \over B_1 - B_2} \left( {y \over y_L} \right)^{B_2} \right] \left( {c \over r} + L \right) {\partial{y_L} \over \partial{L}},
\end{align}
which simplifies to the following after substituting for $c/r - w$ from \eqref{eq:w-hatm-der}:
\begin{align}
{\partial{M_L(w, 0)} \over \partial{L}} &=  - \left[ {1 - B_2 \over B_1 - B_2} \, {y^{B_1} \over y_L^{B_1 - 1}} + {B_1 - 1 \over B_1 - B_2} \, {y^{B_2} \over y^{B_2 -1}} \right] \notag \\
& \quad + \left( {c \over r} + L \right) {(B_1 - 1)(1 - B_2) \over B_1 - B_2} \left[ \left( {y \over y_L} \right)^{B_1} - \left( {y \over y_L} \right)^{B_2} \right]  {\partial{y_L} \over \partial{L}},
\end{align}
which is negative for $y_0 < y < y_L$ if and only if
\begin{align}\label{eq:ineq1}
& \left( {c \over r} + L \right) (B_1 - 1)(1 - B_2) \left[ \left( {y \over y_L} \right)^{B_1} - \left( {y \over y_L} \right)^{B_2} \right]  {\partial{y_L} \over \partial{L}} \notag \\
& \quad - y_L \left[ (1 - B_2) \, \left( {y \over y_L} \right)^{B_1} + (B_1 - 1) \, \left( {y \over y_L} \right)^{B_2} \right] < 0.
\end{align}

To prove the inequality in \eqref{eq:ineq1}, we require an expression for $\partial{y_L}/\partial{L}$.  To that end, note from \eqref{eq:y0yL}, it follows that
\begin{equation}\label{eq:simp}
\left( {y_0 \over y_L} \right)^{B_1 - 1} = {c/r \over c/r + L} \cdot {B_1 - B_2 \over (B_1 - 1)B_2} - {B_1 (1 - B_2) \over (B_1 - 1) B_2} \left( {y_0 \over y_L} \right)^{B_2 - 1}.
\end{equation}
>From \eqref{eq:y0yL}, after differentiating with respect to $L$ and simplifying, we also obtain
\begin{equation}
{ \partial{(y_0/y_L)} \over \partial{L}} = - {c/r \over (c/r + L)^2} \cdot {1 \over 1 - B_2} \cdot {y_0/y_L \over B_1 (y_0/y_L)^{B_1 - 1} - {c/r \over c/r + L}}.
\end{equation}
Next, from \eqref{eq:y0} and from \eqref{eq:simp}, it follows that
\begin{equation}\label{eq:simp1}
{1 \over y_0} = \la \left[ {B_1 - B_2 \over B_1 B_2} \cdot {c \over r} - {1 - B_2 \over B_2} \left( {c \over r} + L \right) \left( {y_0 \over y_L} \right)^{B_1 - 1} \right].
\end{equation}
After differentiating with respect to $L$ and simplifying, we obtain
\begin{equation}
{1 \over y_0^2} {\partial{y_0} \over \partial{L}} = {\la \over B_2} \left( {y_0 \over y_L} \right)^{B_1 - 1} {B_1 (1 - B_2) (y_0/ y_L)^{B_1 - 1}  - (B_1 - B_2) {c/r \over c/r + L} \over B_1 (y_0/y_L)^{B_1 - 1}  - {c/r \over c/r + L}}.
\end{equation}
It follows that
\begin{align}
{\partial{y_L} \over \partial{L}} &= {1 \over y_0/y_L} {\partial{y_0} \over \partial{L}} - y_L {\partial{(y_0/y_L)}/\partial{L} \over y_0/y_L} \notag \\
& = y_0 \, y_L \, {\la \over B_2} \left( {y_0 \over y_L} \right)^{B_1 - 1} {B_1 (1 - B_2) (y_0/ y_L)^{B_1 - 1}  - (B_1 - B_2) {c/r \over c/r + L} \over B_1 (y_0/ y_L)^{B_1 - 1}  - {c/r \over c/r + L}} \notag \\
& \quad + y_L \, {c/r \over (c/r + L)^2} \cdot {1 \over 1 - B_2} \cdot {y_0/y_L \over B_1 (y_0/y_L)^{B_1 - 1} - {c/r \over c/r + L}}.
\end{align}
Substitute for $(y_0/y_L)^{B_1 - 1}$ from \eqref{eq:simp1} into this expression and simplify to obtain
\begin{equation}\label{eq:simp2}
{\partial{y_L} \over \partial{L}} = {y_L \over c/r + L} \left( -1 + {B_2 (c/r) \over (B_1 - 1)(c/r) - B_1 \, B_2/(\la y_0) }\right).
\end{equation}

Then, from \eqref{eq:ineq1} and \eqref{eq:simp2}, $\partial{M_L(w, 0)}/\partial{L} < 0$ if and only if
\begin{align}\label{eq:ineq2}
&(B_1 - 1)(1 - B_2) \left[ \left( {y \over y_L} \right)^{B_1} - \left( {y \over y_L} \right)^{B_2} \right] \, \left[ -1 + {B_2 (c/r) \over (B_1 - 1)(c/r) - B_1 \, B_2/(\la y_0) } \right]  \notag \\
& \quad -  \left[ (1 - B_2) \, \left( {y \over y_L} \right)^{B_1} + (B_1 - 1) \, \left( {y \over y_L} \right)^{B_2} \right] < 0,  \qquad y_0 < y < y_L.
\end{align}
It is straightforward to show that the left-hand side of \eqref{eq:ineq2} is decreasing with respect to $y$; therefore, it is enough to prove \eqref{eq:ineq2} for $y = y_0$.  By substituting for $(y_0/y_L)^{B_1 - 1}$ and $(y_0/y_L)^{B_2 -1}$ in terms of $y_0$, and simplifying the result, we obtain the following inequality:
\begin{equation}
- {B_1 - B_2 \over B_1} \, {c \over r} + {B_2 \over \la y_0} < 0,
\end{equation}
which is clearly true.  Thus, we have shown that $M_L$ decreases as $L$ increases. \hfill $\square$ \\
\\

\noindent \textbf{Acknowledgements.} We thank Mario Boudalha-Ghossoub for suggesting this problem. We also would like to thank the other participants of the Department of Statistics Seminar at the University of Michigan and the two anonymous referees for their comments which helped us improve our paper.

\bibliographystyle{dcu}

\end{document}